\def\llncs{0}
\def\fullpage{1}
\def\anonymous{0}
\def\authnote{1}
\def\notxfont{0}
\def\submission{0}
\def\cameraready{0}

\ifnum\submission=1
\def\anonymous{1}
\def\llncs{1}
\def\authnote{0}
\else
\fi

\ifnum\cameraready=1
\def\submission{1}
\def\llncs{1}
\def\authnote{0}
\def\anonymous{0}
\else
\fi

\ifnum\anonymous=1
\def\llncs{1}
\def\authnote{0}
\else
\fi

\ifnum\llncs=1
	\documentclass[envcountsect,a4paper,runningheads,10pt]{llncs}
\else
	\documentclass[letterpaper,hmargin=1.05in,vmargin=1.05in]{article}
			\ifnum\fullpage=1
		\usepackage{fullpage}
		\fi
\fi

\usepackage[%
  colorlinks=true,
  citecolor=blue,
  pagebackref=true
]{hyperref}

\usepackage{amsmath, amsfonts, amssymb, mathtools,amscd}

\usepackage{amsthm}

\usepackage{tabularx}
\usepackage{fancyhdr}
\usepackage{lmodern}
\usepackage[T1]{fontenc}
\usepackage[utf8]{inputenc}

\usepackage{arydshln} 
\usepackage{url}
\usepackage{ifthen}
\usepackage{bm}
\usepackage{multirow}
\usepackage[dvips]{graphicx}
\usepackage[usenames]{color}
\usepackage{xcolor,colortbl} 
\usepackage{threeparttable}
\usepackage{comment}
\usepackage{paralist,verbatim}
\usepackage{cases}
\usepackage{booktabs}
\usepackage{braket}
\usepackage{cancel} 
\usepackage{ascmac} 
\usepackage{framed}
\usepackage{physics}
\usepackage{authblk}
\usepackage{pifont}
\usepackage{autonum}

\definecolor{darkblue}{rgb}{0,0,0.6}
\definecolor{darkgreen}{rgb}{0,0.5,0}
\definecolor{maroon}{rgb}{0.5,0.1,0.1}
\definecolor{dpurple}{rgb}{0.2,0,0.65}

\usepackage[capitalise,noabbrev]{cleveref}
\usepackage[absolute]{textpos}
\usepackage[final]{microtype}
\usepackage[absolute]{textpos}
\usepackage{everypage}

\newtheoremstyle{thicktheorem}%
{\topsep}
{\topsep}
{\itshape}{}%
{\bfseries}%
{.}
{ }%
{\thmname{#1}\thmnumber{ #2}%
		\thmnote{ (#3)}%
}

\newtheoremstyle{remark}
{\topsep}
{\topsep}
	{}
	{}
	{}
	{.}
	{ }
	{\textit{\thmname{#1}}\thmnumber{ #2}
			\thmnote{ (#3)}%
	}

\ifnum\llncs=0
	\theoremstyle{thicktheorem}
	\newtheorem{theorem}{Theorem}[section]
	\newtheorem{lemma}[theorem]{Lemma}
	\newtheorem{corollary}[theorem]{Corollary}
	\newtheorem{proposition}[theorem]{Proposition}
	\newtheorem{definition}[theorem]{Definition}
	\newtheorem{game}[theorem]{Game}

	\theoremstyle{remark}
	
	\newtheorem{remark}[theorem]{Remark}

\else
\fi

	\crefname{theorem}{Theorem}{Theorems}
	\crefname{assumption}{Assumption}{Assumptions}
	\crefname{construction}{Construction}{Constructions}
	\crefname{corollary}{Corollary}{Corollaries}
	\crefname{conjecture}{Conjecture}{Conjectures}
	\crefname{definition}{Definition}{Definitions}
	\crefname{exmaple}{Example}{Examples}
	\crefname{experiment}{Experiment}{Experiments}
	\crefname{counterexample}{Counterexample}{Counterexamples}
	\crefname{lemma}{Lemma}{Lemmata}
	\crefname{observation}{Observation}{Observations}
	\crefname{proposition}{Proposition}{Propositions}
	\crefname{remark}{Remark}{Remarks}
	\crefname{claim}{Claim}{Claims}
	\crefname{fact}{Fact}{Facts}
	\crefname{note}{Note}{Notes}

\ifnum\llncs=1
 \crefname{appendix}{App.}{Appendices}
 \crefname{section}{Sec.}{Sections}
\else
\fi

\ifnum\llncs=1
\renewcommand*{\backref}[1]{}
\else
	\renewcommand*{\backref}[1]{(Cited on page~#1.)}
	\ifnum\notxfont=1
	\else
		\usepackage{newtxtext}
	\fi
\fi

\ifnum\authnote=0  
\newcommand{\mor}[1]{}
\newcommand{\taiga}[1]{}
\newcommand{\ryo}[1]{}
\newcommand{\takashi}[1]{}

\else
\newcommand{\mor}[1]{$\ll$\textsf{\color{red} Tomoyuki: { #1}}$\gg$}
\newcommand{\taiga}[1]{$\ll$\textsf{\color{magenta} Taiga: { #1}}$\gg$}
\newcommand{\takashi}[1]{$\ll$\textsf{\color{orange} Takashi: { #1}}$\gg$}
\newcommand{\ryo}[1]{$\ll$\textsf{\color{darkgreen} Ryo: { #1}}$\gg$}

\fi

\newcommand{\StateGen}{\mathsf{StateGen}}

\newcommand{\Samp}{\algo{Samp}}

\newcommand{\la}{\leftarrow}
\newcommand{\ra}{\rightarrow}

\newcommand{\seteq}{\coloneqq}

\newcommand{\cA}{\mathcal{A}}
\newcommand{\cB}{\mathcal{B}}
\newcommand{\cC}{\mathcal{C}}
\newcommand{\cD}{\mathcal{D}}

\newcommand{\cL}{\mathcal{L}}
\newcommand{\cM}{\mathcal{M}}

\newcommand{\cQ}{\mathcal{Q}}
\newcommand{\cR}{\mathcal{R}}
\newcommand{\cS}{\mathcal{S}}
\newcommand{\cT}{\mathcal{T}}

\newcommand{\cX}{\mathcal{X}}

\newcommand{\cZ}{\mathcal{Z}}

\def\makeuppercase#1{
\expandafter\newcommand\csname tl#1\endcsname{\widetilde{#1}}
}

\def\makelowercase#1{
\expandafter\newcommand\csname tl#1\endcsname{\widetilde{#1}}
}

\newcommand{\N}{\mathbb{N}}

\newcommand{\R}{\mathbb{R}}

\newcommand{\secp}{\lambda}

\newcommand*{\algo}[1]{\ensuremath{\mathsf{#1}}}

\newcommand{\compclass}[1]{\textbf{\textrm{#1}}}

\newenvironment{boxfig}[2]{\begin{figure}[#1]\fbox{\begin{minipage}{0.97\linewidth}
                        \vspace{0.2em}
                        \makebox[0.025\linewidth]{}
                        \begin{minipage}{0.95\linewidth}
            {{
                        #2 }}
                        \end{minipage}
                        \vspace{0.2em}
                        \end{minipage}}}{\end{figure}}

\newcommand{\bit}{\{0,1\}}

\newcommand{\keygen}{\algo{KeyGen}}

\newcommand{\Vrfy}{\algo{Vrfy}}

\newcommand{\negl}{{\mathsf{negl}}}

\newcommand{\poly}{{\mathrm{poly}}}

\ifnum\llncs=1
\let\oldvec\vec
\let\vec\oldvec
\makeatletter
\renewcommand*\l@author[2]{}
\renewcommand*\l@title[2]{}
\makeatletter

\else
\fi    

\theoremstyle{remark}

\title{
\textbf{
Computational Complexity of Learning Efficiently Generatable Pure States
} 
}

\begin{document}

\ifnum\anonymous=1
\author{\empty}\institute{\empty}
\else
%
%
\ifnum\llncs=1
\index{Taiga, Hiroka}
\author{
	Taiga Hiroka\inst{1} 
}
\institute{
	Yukawa Institute for Theoretical Physics, Kyoto University, Japan  \and NTT Corporation, Tokyo, Japan
}
\else
%
%

\author[1]{Taiga Hiroka}
\author[2]{Min-Hsiu Hsieh}
\affil[1]{{\small Yukawa Institute for Theoretical Physics, Kyoto University, Kyoto, Japan}\authorcr{\small taiga.hiroka@yukawa.kyoto-u.ac.jp}}
\affil[2]{{\small Hon Hai Research Institute, Taipei, Taiwan}\authorcr{\small min-hsiu.hsieh@foxconn.com}}
\renewcommand\Authands{, }
\fi 
\fi

\ifnum\llncs=1
\date{}
\else
\date{\today}
\fi

\maketitle
\ifnum\llncs=0
\thispagestyle{fancy}
\rhead{YITP-24-126}
\else
\fi

\begin{abstract}

Understanding the computational complexity of learning efficient classical programs in various learning models has been a fundamental and important question in classical computational learning theory.
In this work, we study the computational complexity of quantum state learning, which can be seen as a quantum generalization of distributional learning introduced by Kearns et.al [STOC94].
Previous works by Chung and Lin [TQC21], and Bădescu and O’Donnell [STOC21] study the sample complexity of the quantum state learning and show that polynomial copies are sufficient if unknown quantum states are promised efficiently generatable. 
However, their algorithms are inefficient, and the computational complexity of this learning problem remains unresolved.

In this work, we study the computational complexity of quantum state learning when the states are promised to be efficiently generatable.
We show that if unknown quantum states are promised to be pure states and efficiently generateable, then there exists a quantum polynomial time algorithm $\cA$ and a language $\cL\in\mathbf{PP}$ such that $\cA^{\cL}$ can learn its classical description.
We also observe the connection between the hardness of learning quantum states and quantum cryptography.
We show that the existence of one-way state generators with pure state outputs is equivalent to the average-case hardness of learning pure states.
Additionally, we show that the existence of EFI implies the average-case hardness of learning mixed states.

\end{abstract}

\ifnum\cameraready=1
\else
\ifnum\llncs=1
\else
\newpage
  \setcounter{tocdepth}{2}      
  \setcounter{secnumdepth}{2}   
  \setcounter{page}{1}          
  \thispagestyle{empty}
  \clearpage
\fi
\fi

\section{Introduction}


Since Valiant~\cite{ACM:Valaint84} formalized the notion of PAC learning, understanding the computational complexity of learning efficient classical programs in various learning models has been a fundamental and important question in classical computational learning theory.
It is shown that we can learn arbitrary efficient classical programs in the PAC learning model if we can solve an arbitrary problem in $\mathbf{NP}$~\cite{IPL:BEHW87,TCS:BP92,ML:Schapire90}.
The $\mathbf{NP}$ hardness of learning efficient classical programs in the \emph{proper} PAC learning model is also shown~\cite{ACM:PV88}.
While showing the $\mathbf{NP}$ hardness of \emph{improper} PAC learning has been an open question, several works~\cite{STOC:DLS14,COLT:DS16,COLT:Vadhan17,COLT:DV21} show its hardness assuming the specific average-case hardness assumptions of $\mathbf{NP}$ (See also \cite{FOCS:HN21,FOCS:Hir22}).
Understanding the computational complexity of learning efficient classical programs is also crucial for constructing classical cryptographic primitives.
In fact, it is known that the existence of fundamental classical cryptographic primitives such as one-way functions (OWFs) is equivalent to the average-case hardness of learning efficient classical programs in the various learning models~\cite{FOCS:ImpLub90,C:BFKL93,ACM:NR06,CCC:OS17,ACM:NE19,FOCS:HirNana23}.
Furthermore, it has been shown that assuming the hardness of specific classical learning problems, such as the Learning with Errors (LWE) problem, various useful classical cryptographic protocols can be constructed~\cite{JACM:Regev09}.
Given the importance of understanding the computational complexity of learning efficient classical programs, it is an interesting direction to study computational complexity of learning efficient quantum processes in the various learning models.

In this work, we study the computational complexity of quantum state learning, which can be seen as a quantum generalization of distributional learning introduced by \cite{ACM:KMRRSS94}.
In quantum state learning, a learner is given many copies of an unknown quantum state which is promised within some family of quantum states, and the task is to find a hypothetical quantum circuit $h$ whose output is close to the given unknown quantum state.
When there is no promise on the input states,
it is known that an exponential number of copies relative to the number of qubits is required~\cite{ACM:OJ16,ACM:HHJWY16}.
On the other hand, because many quantum states of interest can be efficiently generated, it is natural to ask the sample and computational complexities of quantum state learning with the promise that an unknown quantum state is efficiently generatable.
In previous works, Chung and Lin~\cite{TQC:KH21}, and B\v{a}descu and O'Donell~\cite{STOC:CR21} independently showed that polynomial copies are sufficient when an unknown quantum state is promised to be efficiently generatable.
However, their algorithms are not time-efficient, and it is natural to ask the computational complexity of the learning task.

Currently, it is known that if there exist pseudorandom state generators (PRSGs)~\cite{C:JiLiuSon18}, then there exist pure states that can be generated by quantum polynomial-time algorithms, but hard to learn its description~\cite{arXiv:ZLKQHC23,arXiv:CLS24}.
This implies that it seems difficult to learn arbitrary quantum states generated by quantum polynomial-time algorithms even if we can query $\mathbf{QMA}$ oracle because it is known that PRSGs exist relative to an oracle $\mathcal{O}$ such that $\mathbf{BQP}^{\mathcal{O}}=\mathbf{QMA}^{\mathcal{O}}$~\cite{TQC:Kre21}.
On the other hand, we do not know whether a QPT algorithm can learn quantum states generated by quantum polynomial-time algorithms even if the algorithm can query $\mathbf{PSAPCE}$ oracle.
In particular, in quantum state learning, instances are themselves quantum states, and hence it is unclear how to reduce quantum state learning to traditional classical languages, where instances are represented by classical strings.
Therefore, we believe that the following is an interesting open question:
\begin{center}
\emph{Which oracle does help us to learn quantum states generated by an arbitrary quantum polynomial-time algorithm?}
\end{center}

We also ask connection between quantum state learning and quantum cryptography.
In the classical case, there exists an equivalence between the average-case hardness of learning in the various learning models and the existence of cryptography~\cite{FOCS:ImpLub90,C:BFKL93,ACM:NR06,CCC:OS17,ACM:NE19,FOCS:HirNana23}.
For example, it is known that the average-case hardness of distirutional learning is equivalent to the existence of OWFs~\cite{FOCS:HirNana23}.
It is natural to wonder whether a similar relationship holds in the quantum case.
Therefore, we ask
\begin{center}
\emph{Is the average-case hardness of learning quantum states equivalent to the existence of quantum cryptographic primitives?}
\end{center}

\color{black}

\subsection{Our Results}

In this work, we study two questions above.
Our contributions are summarized as follows.

\begin{table}[h!]
\centering
\begin{tabularx}{\textwidth}{|X|X|X|}
\toprule
\textbf{Types of Problem} &  \multicolumn{2}{c|}{\textbf{Computational Complexity}} \\
\cmidrule(lr){2-3} 
 & \textbf{Upper Bound} & \textbf{Lower Bound} \\
\midrule
Learning Pure States generated in Polynomial Time & 
$\mathbf{PP}$~(\cref{thm:Learning_w_Oracle}) & OWSGs with Pure State Outputs \\
\hline
Average-Case Learning Pure States generated in Polynomial Time & OWSGs with Pure State Outputs~(\cref{thm:owsg_a_h_learn}) & OWSGs with Pure State Outputs~(\cref{thm:owsg_a_h_learn}) \\
\hline
Learning Mixed State generated in Polynomial Time & N/A & EFI~(\cref{thm:Hardness_from_EFI}) \\
\bottomrule
\end{tabularx}
\caption{Computational Complexities in Quantum State Learning}
\label{table}
\begin{tablenotes}
    \item 
    In column "Computational Complexity", the upper bound means that the learning tasks can be achieved with the listed class of algorithms, and
    the lower bound means that the listed class of algorithms is necessary to achieve the learning task.
    $\mathbf{PP}$ stands for QPT algorithm with PP oracle, OWSGs with pure state outputs stand for QPT algorithm which can query an algorithm that breaks any OWSGs with pure state outputs, and EFI stands for QPT algorithm which can query an algorithm that breaks any EFI.
\end{tablenotes}
\end{table}

\begin{enumerate}
    \item 
    In \cref{sec:learning_model}, we formally define a quantum state learning model,
    which can be seen as a quantum generalization of the distributional learning model~\cite{ACM:KMRRSS94}.
    The distributional learning model can be categorized into two flavors: proper distributional learning and improper distributional learning.
    In this work, we focus on the proper one and generalize it to a quantum setting.
    In our proper quantum state learning model, as a setup, a learner $\cA$ is given a description of quantum polynomial-time algorithm $\cC$, which takes as input $1^n$ and $x\in\cX_n=\bit^{\poly(n)}$ and outputs a quantum state $\psi_x$.
    Then, it receives polynomially many copies of $\psi_x$ generated by $\cC$.
    We say that the leaner $\cA$ can properly learn $\cC$ if, \emph{for all $x\in\cX_n$}, $\epsilon$ and $\delta$, it can find $y\in\cX_n$ such that $\mathsf{TD}(\psi_x,\psi_y)\leq1/ \epsilon$ with probability at least $1-1/\delta$, where $\psi_y\la\cC(1^n,y)$\footnote{
    The other possible learning model is improper quantum state learning.
    In the improper quantum state learning model, a learner's task is to find an arbitrary quantum algorithm $h$ whose output is close to $\ket{\psi_n}$ instead of finding $y\in\cX_n$ such that $\ket{\psi_y}$ is close to $\ket{\psi_x}$.
    Remark that if we can properly learn $\cC$, then we can also improperly learn $\cC$. 
    On the other hand, the other direction may not hold.
}.

\item
In \cref{sec:quantum_state_learning}, we show that, for all quantum polynomial-time algorithms $\cC$ with pure state outputs, there exists a quantum polynomial-time algorithm $\cA$ and a language $\cL\in\mathbf{PP}$ such that $\cA^{\cL}$ can properly learn $\cC$ given $\poly(n,\log(\abs{\cX_n}),\epsilon,\log(\delta) )$-copies of $\ket{\psi_x}$ with unknown $x$~\footnote{Our QPT algorithm with $\mathbf{PP}$ oracle can be applied to improper learning.}.
This implies that, if a novel algorithm is found to solve any $\mathbf{PP}$ problems, then it will be useful to learn arbitrary pure states generated by any QPT algorithms.
Let us stress that our learning algorithm with $\mathbf{PP}$ oracle cannot be applied to learning mixed state generated by any QPT algorithms.

\item In \cref{sec:hardness}, we study the connection between quantum state learning and quantum cryptography.
We observe that the average-case hardness of learning pure states in quantum state learning model is equivalent to the existence of one-way state generators (OWSGs) with pure state outputs~\cite{Crypto:MY22}.
In OWSGs with pure state outputs, the adversary receives many copies of $\ket{\psi_x}$, where $x$ is randomly sampled and $\ket{\psi_x}$ is prepared by some QPT algorithm $\cC(x)$.
What the adversary needs to do is to find $h$ such that the measurement $\{\ket{\psi_h}\bra{\psi_h},I-\ket{\psi_h}\bra{\psi_h}\}$ on $\ket{\psi_x}$ yields $\ket{\psi_h}\bra{\psi_h}$ with high probability.
This is very similar to what the learner needs to do in our learning model.
We observe that we can break arbitrary OWSGs with pure state outputs if and only if we can learn pure states in the average-case version of our learning model.
Intuitively, in average-case learning, the learner needs to find $h$ such that $\ket{\psi_h}$ is close to $\ket{\psi_x}$ \emph{for a large fraction of $\cX_n$ instead of for all $x\in\cX_n$}.
Therefore, the learning requirement is relaxed compared to worst-case learning.

We also observe that the existence of EFI~\cite{Asia:Yan22,ITCS:BCQ23} implies that there are mixed states generated by quantum polynomial-time algorithms, which is hard to properly learn.
Because \cite{arXiv:LMW23} conjectured that EFI may exist relative to a random oracle and any classical oracle, our observation implies that any classical oracles seem not help to properly learn arbitrary mixed states generated by quantum polynomial-time algorithms~\footnote{Remark that the argument cannot be applied to the hardness of improperly mixed state learning, and thus there exists the possibility that we can improperly learn mixed states by querying classical oracles.}.
Therefore, this implies that QPT algorithms with classical oracles do not seem applicable to learning any efficiently generatable mixed states.
\end{enumerate}

The upper and lower bounds of the computational complexity of the quantum state learning problem are summarized in the \cref{table}.

\subsection{Technical Overview}
In this section, we explain how a quantum polynomial-time algorithm with $\mathbf{PP}$-oracle can learn pure states generated by a quantum polynomial-time algorithm.

\paragraph{\emph{Our Learning Algorithm.}}
Our technical contribution is, given $\mathbf{PP}$ oracle, showing how to efficiently implement Chung and Lin's algorithm~(Given in Section4 in \cite{TQC:KH21}).
First, let us briefly recall their strategy in our language.
\begin{description}
    \item[Description of CL algorithm $\cA$:] $ $
    \begin{enumerate}
        \item As a setup, $\cA$ receives a quantum algorithm $\cC$ that takes $x\in\cX_n=\bit^{\poly(n)}$ and outputs $\ket{\psi_x}$.
        \item $\cA$ receives $\{\ket{\psi_{x}}_i\}_{i\in[T]}$ with sufficiently large $T$ and an unknown $x\in\cX_n$.
        \item For all $i\in[T]$, $\cA$ measures $\ket{\psi_x}_i$ with Haar random POVM measurements $\cM$, and obtains $z_i$.
        Note that this process is possibly inefficient.
        Let $\cC^*(\cdot)$ be a quantum algorithm that takes $x$, generates $\ket{\psi_x}$, measure it with $\cM$, and outputs the measurement outcome.
        \item $\cA$ outputs $h\in\cX_n$ such that 
        \begin{align}
\max_{x\in\cX_n}\{\Pi_{i\in[T]}\Pr[z_i\la\cC^*(x)]\}=\Pi_{i\in[T]}\Pr[z_i\la\cC^*(h)].
        \end{align}
        Note that this step is possibly inefficient.
    \end{enumerate}
\end{description}
Let us explain an intuitive reason why this algorithm works.
First, it is known that Haar random measurements do not decrease the trace distance between quantum states if measured states are pure states~\cite{Pranab06}.
From this property, the measurements in the third step of $\cA$ does not decrease the trace distance between quantum states.
In other words, there is a constant $c$ such that, for all $a,b\in\cX_n $,
\begin{align}
    \norm{\cC^*(a)-\cC^*(b)}_1\geq c \norm{\ket{\psi_a}-\ket{\psi_b}}_1.
\end{align}
This implies that, for finding $h$ such that $\norm{\ket{\psi_h}-\ket{\psi_x}}_1$ is small,
it is sufficient to find a $h\in\cX_n$, such that $\norm{\cC^*(h)-\cC^*(x)}_1$ is small, given polynomially many copies of $z\la\cC^*(x)$.
In the fourth step, $\cA$ finds a maximum-likelihood $h\in\cX_n$ such that $\cC^*(h)$ is most likely to output $\{z_i\}_{i\in[T]}$.
It is natural to expect that a maximum-likelihood algorithm $\cC^*(h)$ seems to approximate $\cC^*(x)$ well with high probability.
\cite{TQC:KH21} shows that if $T\geq O(\epsilon^2\cdot(\log(\abs{\cX_n})+\log(\delta)))$, then we have $\norm{\cC^*(h)-\cC^*(x)}_1\leq 1/\epsilon$ with probability $1-1/\delta$, where the probability is taken over $\{z_i\}_{i\in[T]}\la \cC^*(x)$.

Although their algorithm learns pure states, their algorithm does not work efficiently.
The first bottleneck is the third step of $\cA$.
In the third step, $\cA$ needs Haar random measurements, which is hard to implement.
Fortunately, it is shown that we have the same property if we implement $4$-design measurements instead of Haar random measurements~\cite{CCC:AE07}.
Therefore, we can efficiently implement the third step by running $4$-design measurements instead of Haar random measurements.
The second bottleneck is the fourth step of $\cA$.
In the fourth step, given $\{z_i\}_{i\in[T]}$, $\cA$ needs to find a maximum-likelihood $h\in\cX_n$ such that 
\begin{align}
\max_{x\in\cX_n}\{\Pi_{i\in[T]}\Pr[z_i\la\cC^*(x)]\}=\Pi_{i\in[T]}\Pr[z_i\la\cC^*(h)],
\end{align}
which we currently do not know how to implement efficiently.
In our work, we show that we can achieve an approximated version of the task by querying $\mathbf{PP}$ oracle. 
More precisely, we show that for arbitrary quantum polynomial-time algorithms $\cC^*$ with classical inputs and classical outputs, there exists a language $\cL\in\mathbf{PP}$ and an oracle-aided PPT algorithm $\cA^{\cL}$ that can find an approximated maximum-likelihood $h\in\cX_n$ such that
\begin{align}
\frac{\max_{x\in\cX_n}\left\{\Pi_{i\in[T]}\Pr[z_i\la\cC^*(x)]\right\}}{\left(1+1/\epsilon(n)\right)}\leq \Pi_{i\in[T]}\Pr[z_i\la\cC^*(h)]\leq \max_{x\in\cX_n}\left\{\Pi_{i\in[T]}\Pr[z_i\la\cC^*(x)]\right\},
\end{align}
where $\epsilon$ is arbitrary polynomial.

A subtle but important point is that an ``approximated'' maximum-likelihood algorithm $\cC^*(h)$ might not statistically approximate $\cC^*(x)$ even if we take $T$ sufficiently large.
Actually, we do not know how to prove an approximated maximum-likelihood algorithm $\cC^*(h)$ approximates $\cC^*(x)$ for sufficiently large $T$.
To overcome the issue, inspired by \cite{ML:AW92}, we slightly modify the learning algorithm $\cA$ as follows.
In a modified learning algorithm $\cB$, instead of $\cC^*(\cdot)$,
$\cB$ considers a quantum algorithm $\cC^*_{1/2}(\cdot)$, where $\cC^*_{1/2}(\cdot)$ takes as input $x\in\cX_n$, and outputs the output of $\cC^*(x)$ with probability $1/2$, and outputs a uniformly random string $z$ with probability $1/2$.
Given many copies of $\ket{\psi_x}_i$, $\cB$ measures $\ket{\psi_x}_i$ with Haar random measurement and sets the measurement outcome as $z_i$ with probability $1/2$, and uniform random string as $z_i$ with probability $1/2$. 
$\cB$ outputs $h$ such that
\begin{align}
\frac{\max_{x\in\cX_n}\left\{\Pi_{i\in[T]}\Pr[z_i\la\cC^*_{1/2}(x)]\right\}}{\left(1+1/\epsilon(n)\right)}\leq \Pi_{i\in[T]}\Pr[z_i\la\cC^*_{1/2}(h)]\leq\max_{x\in\cX_n}\left\{\Pi_{i\in[T]}\Pr[z_i\la\cC^*_{1/2}(x)]\right\}
\end{align}
This modification makes it possible to prove that for sufficiently large $T$,  $\norm{\cC^*_{1/2}(x)-\cC^*_{1/2}(h)}_1$ is small with high probability.
Furthermore, we have
\begin{align}
   \norm{\cC^*_{1/2}(x)-\cC^*_{1/2}(h)}_1= \frac{1}{2} \norm{\cC^*(x)-\cC^*(h)}_1\geq c \norm{\ket{\psi_x}-\ket{\psi_h}}_1.
\end{align}
This means that $\ket{\psi_h}$ statistically approximates $\ket{\psi_x}$.
Therefore, the remaining part is showing how to find an approximated maximum-likelihood algorithm $\cC^*(h)$ given $\mathbf{PP}$ oracle.

\paragraph{\emph{Probabilistic Polynomial-Time Algorithm with PP oracle for Quantum Maximum-Likelihood Problem.}}
In the rest of the section, we explain how to solve the following quantum maximum-likelihood problem using $\mathbf{PP}$ oracle.
\begin{description}
    \item[Quantum Maximum-Likelihood Problem (QMLH).]$ $ 
    \begin{itemize}
        \item[Input:]$ $
        A classical string $z$ and a quantum algorithm $\cC$, which takes as input $1^n$ and a classical string $x\in\cX_n=\bit^{\poly(n)}$, and outputs a classical string $z\in\cZ_n$. (In the following, we often omit $1^n$ of $z\la\cC(1^n,x)$.)
        \item[Output:]
        Output $h\in\cX_n$ such that
        \begin{align}
            \frac{\max_{x\in\cX_n}\{\Pr[z\la\cC(x)]\}}{\left(1+1/\epsilon(n)\right)}\leq \Pr[z\la\cC(h)]\leq \max_{x\in\cX_n}\{\Pr[z\la\cC(x)]\}.            
        \end{align}
    \end{itemize}
\end{description}
As explained later, a PPT algorithm $\cA$ with $\mathbf{PP}$ oracle can exactly simulate an arbitrary probability distribution generated by a quantum polynomial-time algorithm even if it does post-selection.
Therefore, it is sufficient to construct a quantum algorithm $\widehat{\cQ}[z]$ with post-selection, which solves the quantum maximum likelihood problem.
Let us explain how $\widehat{\cQ}[z]$ works.
\begin{enumerate}
    \item For sufficiently large $T$, $\widehat{\cQ}[z]$ generates the following quantum state
    \begin{align}
        U_Q\ket{0}_{W,Y_1V_1\cdots, Y_TV_T}&\seteq\frac{1}{\sqrt{\abs{\cX_n}}}\sum_{x\in\cX_n}\ket{x}_W\bigotimes_{i\in[T]} \ket{\cC(x)}_{Y_iV_i}\\
        &=
        \frac{1}{\sqrt{\abs{\cX_n}}}\sum_{x\in\cX_n}\ket{x}_W\bigotimes_{i\in[T]} \left(\sum_{y\in\cZ_n} \sqrt{\Pr[y\la \cC(x)]} \ket{y}_{Y_i}\ket{\phi_{x,y}}_{V_i} \right).
    \end{align}
    Here, we write $\ket{\cC(x)}_{Y_iV_i}$ to mean a quantum state generated by a quantum algorithm $\cC(x)$ before measuring the state with the computational basis, where the $Y_i$ register corresponds to the outcome register and the $V_i$ register corresponds to the auxiliary register.
    \item It post-selects all of $Y_1\cdots Y_T$ registers to $z$.
    Let $\ket{\psi[z]}_{WY_1V_1\cdots Y_TV_T}$ be the post-selected quantum state.
    \item 
    It measures the $W$ register of $\ket{\psi[z]}_{WY_1V_1\cdots Y_TV_T}$ in the computational basis, and outputs the measurement outcome $h$.
\end{enumerate}
From the construction of $\widehat{\cQ}[z]$,$\ket{\psi[z]}_{WY_1V_1\cdots Y_TV_T}$ satisfies
\begin{align}
    \ket{\psi[z]}_{WY_1V_1\cdots Y_TV_T}=\frac{1}{\sqrt{\sum_{x\in\cX_n}(\Pr[z\la\cC(x)]^T)}} \sum_{h\in\cX_n}\sqrt{\Pr[z\la\cC(h) ]^T}\ket{h}_W\bigotimes_{i\in[T]}\ket{z}_{Y_i}\ket{\phi_{h,z}}_{V_i}.
\end{align}
Therefore, it is natural to expect that, if we measure the $W$ register of $\ket{\psi[z]}_{WY_1V_1\cdots Y_TV_T}$, then, we obtain a measurement outcome $h$ such that $\Pr[z\la \cC(h)]$ is high.
We can prove the intuition by taking $T$ sufficiently large as follows.
If $h$ is not a good answer for the quantum maximum likelihood problem (i.e. $\Pr[z\la\cC(h)]<\frac{\max_{x\in\cX_n}\Pr[z\la\cC(x)]}{\left(1+1/\epsilon(n)\right)}$), then 
\begin{align}
    \Pr[h\la \widehat{\cQ}[z]]=\frac{\Pr[z\la\cC(h) ]^T}{\sum_{x\in\cX_n}\Pr[z\la\cC(x)]^T}\leq \frac{\Pr[z\la\cC(h) ]^T}{\max_{x\in\cX_n}\Pr[z\la\cC(x)]^T}<\left(\frac{1}{1+1/\epsilon} \right)^T.
\end{align}
This implies that 
\begin{align}
  \Pr_{h\la\widehat{\cQ}[z]}\left[ \Pr[z\la\cC(h)]<\frac{\max_{x\in\cX_n}\Pr[z\la\cC(x)]}{\left(1+1/\epsilon(n)\right)}\right] \leq \abs{\cX_n}\cdot\left(\frac{1}{1+1/\epsilon} \right)^T.
\end{align}
\color{black}
Therefore, by taking $T$ sufficiently large, $\widehat{\cQ}[z]$ outputs $h$ such that
\begin{align}
   \frac{\max_{x\in\cX_n}\Pr[z\la\cC(x)]}{1+1/\epsilon}\leq \Pr[z\la\cC(h)]\leq \max_{x\in\cX_n}\Pr[z\la\cC(x)]
\end{align}
with high probability.

Although $\widehat{\cQ}[z]$ solves the quantum maximum likelihood problem with high probability, it is possibly inefficient since it does post-selection.
In the following, we explain how a PPT algorithm $\cA$ with $\mathbf{PP}$ oracle exactly simulates the probability distribution of $\widehat{\cQ}[z]$.
As shown in \cite{FR99}, even a deterministic polynomial time algorithm $\cB$ with $\mathbf{PP}$ oracle can compute $\Pr[z\la\cC]$ for an arbitrary QPT algorithm $\cC$ and a classical string $z$.
Our PPT algorithm $\cA^{\mathbf{PP}}$ simulates $\widehat{\cQ}[z]$ by using $\cB^{\mathbf{PP}}$ as follows:
\begin{itemize}
    \item For each $i\in[\abs{x}]$, $\cA^{\mathbf{PP}}$ sequentially works as follows:
    \begin{enumerate}
        \item For each $b\in\bit$, it computes the probability $\Pr[b,x_{i-1},\cdots, x_1,z]$ and $\Pr[x_{i-1},\cdots,x_1,z]$.
        Here, 
        \begin{align}
            \Pr[x_{i-1},\cdots,x_1,z]\seteq \Tr(\left(\ket{x_{i-1}\cdots x_1}\bra{x_{i-1}\cdots x_1}_{W[i-1]}\bigotimes_{i\in[T]}\ket{z}\bra{z}_{Y_i}\right) U_\cQ\ket{0}_{W,Y_1V_1,\cdots ,Y_TV_T}),
        \end{align}
        where $W[i-1]$ is the register in the first $(i-1)$-bits of $W$.
        Note that, in this step, $\cA^{\mathbf{PP}}$ uses $\cB^{\mathbf{PP}}$ given in \cite{FR99}.
        \item 
        Set $x_i\seteq b$ with probability
        \begin{align}
            \frac{\Pr[b,x_{i-1},\cdots,x_1,z ]}{\Pr[x_{i-1},\cdots,x_1,z ]}.
        \end{align}
    \end{enumerate}
    \item Output $h\seteq x_{\abs{x}},\cdots,x_1 $.
\end{itemize}
From the construction of $\cA^{\mathbf{PP}}$, we have
\begin{align}
    \Pr[x\la \cA^{\mathbf{PP}}]\seteq \frac{\Pr[x_\abs{x},\cdots,x_1,z]}{\Pr[z]}.
\end{align}
for all $x\in\cX_n$.
From the definition, $\frac{\Pr[x_\abs{x},\cdots,x_1,z]}{\Pr[z]}$ is exactly equal to the probability that $\widehat{Q}[z]$ outputs $x$.
This means that the output distribution of $\cA^{\mathbf{PP}}$ is equivalent to that of $\widehat{Q}[z]$.

\subsection{More on Related works}\label{sec:related}
\paragraph{On Quantum Cryptography.}
Ji, Liu, and Song~\cite{C:JiLiuSon18} introduce a notion of PRSGs, and show that it can be constructed from OWFs.
Kretchmer~\cite{TQC:Kre21} shows that PRSGs exist relative to an oracle $\mathcal{O}$ such that $\mathbf{BQP}^{\mathcal{O}}=\mathbf{QMA}^{\mathcal{O}}$, and that PRSGs does not exist if its adversary can query $\mathbf{PP}$ oracle.
The subsequent works~\cite{arXiv:ZLKQHC23,arXiv:CLS24} observe that the existence of PRSGs implies that there exists a quantum state, which can be generated by quantum algorithms but hard to learn.
Morimae and Yamakawa~\cite{Crypto:MY22} introduce the notion of OWSGs, and show how to construct them from PRSGs.
In their original definition of OWSGs, the output quantum states are restricted to pure states, and its definition is generalized to mixed states by \cite{Eprint:MY22}.
In this work, we focus on the original pure-state version.
Based on \cite{arxiv:KT23}, Cavalar et. al.~\cite{arXiv:CGGHLP23} shows that OWSGs with pure state outputs can be broken if its adversary can query $\mathbf{PP}$ oracle.
As pointed out in our work, the existence of OWSGs with pure state outputs is equivalent to the average-case hardness of learning pure states generated by a quantum polynomial-time algorithm.
Therefore, to rephrase their results in terms of learning, a quantum polynomial-time algorithm with $\mathbf{PP}$ oracle can learn an arbitrary pure state generated by quantum polynomial-time algorithms \emph{in the average sense}.
Remark that their algorithm works only for average-case learning, and it is unclear whether their QPT algorithm with PP oracle can learn pure states in the worst-case sense.

Brakerski, Canetti, and Qian~\cite{ITCS:BCQ23} introduce the notion of EFI. (See also \cite{Asia:Yan22}).
It is shown that the existence of EFI is equivalent to that of quantum bit commitments, oblivious transfer, and multi-party computation~\cite{EC:GLSV21,C:BCKM21a,Asia:Yan22,ITCS:BCQ23}.
It is conjectured that EFI may exist relative to a random oracle and any classical oracle~\cite{arXiv:LMW23}.

\paragraph{On Classical Learning Theory.}
Valiant~\cite{ACM:Valaint84} introduced the PAC (probabilistically approximately correct) learning model to capture the notion of computationally efficient learning.
It is shown that PAC learning is equivalent to Occam learning~\cite{IPL:BEHW87,TCS:BP92,ML:Schapire90}, which can be formulated as a search problem in $\compclass{NP}$.
Inspired by Valiant's formalization,
Kearns et. al.~\cite{ACM:KMRRSS94} formalizes a distributional learning model.
Note that our quantum state learning model can be seen as a quantum generalization of the distributional learning model.
It is shown that a polynomial time learnability of a class of probabilistic automata $\cC$ with respect to Kullback–Leibler divergence is equivalent to that of the maximum likelihood problem for the same class $\cC$~\cite{ML:AW92}.

Valiant observed that the existence of pseudorandom functions, which is equivalent to the existence of OWFs~\cite{JACM:GolGolMic86,HILL99}, implies the hardness of learning for P/poly in the PAC model~\cite{ACM:Valaint84}.
Impagliazzo and Levin~\cite{FOCS:ImpLub90} initiate the study of the reverse direction, i.e. from the average-case hardness of learning to cryptography.
Subsequent work~\cite{C:BFKL93,ACM:NR06,CCC:OS17,ACM:NE19,FOCS:HirNana23} studies how to construct cryptographic primitives from the average-case hardness of learning in more general learning models.
Among them, our result (i.e. the equivalence between OWSGs and average-case hardness of learning pure states) can be seen as a quantum analog of \cite{FOCS:HirNana23} in the following sense.
In one of their results, they show the equivalence between the existence of OWFs and the average-case hardness of distributional learning~\cite{ACM:KMRRSS94}.
In our work, we observe that the existence of OWSGs with pure state outputs is equivalent to the average case hardness of learning pure states.

\paragraph{More on Quantum Learning.}
In quantum state tomography, a learner is given many independent copies of an unknown $n$-qubit quantum state $\rho$, and outputs a classical description $\widehat{\rho}$ such that $\mathsf{TD}(\rho,\widehat{\rho})\leq 1/\epsilon$.
It is known that a learner needs $O(2^{2n}\cdot\epsilon^2)$-copies if there is no promise on the unknown quantum state~\cite{ACM:OJ16,ACM:HHJWY16}.
Given the negative results, a natural question is which subclass of quantum states can be learned using only polynomially many copies of unknown quantum states.
\cite{TQC:KH21,STOC:CR21} shows that if given quantum states $\rho_k$ are promised contained in some restricted families $\cC$ of quantum states, then a learner can find $\rho_h\in\cC$ such that $\mathsf{TD}(\rho_k,\rho_h)\leq 1/\epsilon$ with probability $1-1/\delta$ given only $\poly(\log(\abs{\cC}),\log(\delta),\epsilon)$-many copies of $\rho_k$.
This means that all quantum states generated by quantum polynomial-time algorithms can be learned in a sample-efficient way.

Beyond sample efficiency, these lines of work (\cite{arXiv:Montanaro17,arXiv:GIKL23A,arXiv:GIKL23B,arXiv:HLBKALM24}) study which subclass of quantum states can be learned in polynomial-time.
It is shown that stabilizer quantum circuits can be learned in a time-efficient way~\cite{arXiv:Montanaro17}.
Grewal, Iyer, Kretschmer, and Liang~\cite{arXiv:GIKL23A,arXiv:GIKL23B} generalize their results and show that low-rank stabilizer pure states can be learned in a time-efficient way.
Huang, Liu, Broughton, Kim, Anshu, Landau, and McClean~\cite{arXiv:HLBKALM24} show that a family of pure states generated by geometrically local shallow quantum circuits can be learned in a time-efficient way.
Compared to their unconditional results, although our learning algorithm needs to query $\mathbf{PP}$ oracle, our learning algorithm can learn a classical description of arbitrary efficiently generatable pure states.

These lines of work (\cite{Aaronson07,Q:CHY16,Quantum:R18,JSM:ACHKN19,Quantum:GL22,arXiv:CHLLWY22,COLT:BJ95, QIP:AR07,CCC:AR17,PRA:GKZ19,arXiv:SSG23,TQC:KH21,STOC:Aaronson18,NP:HKP20,STOC:CR21}) study how to learn quantum processes in other learning models.
Aaronson~\cite{Aaronson07} considers one of quantum generalization of the PAC learning model~\cite{ACM:Valaint84}, and the subsequent works~\cite{Quantum:R18,JSM:ACHKN19,Quantum:GL22,arXiv:CHLLWY22} study the learning model.
In their model, a learner is given examples $(E_i,\Tr(E_i\rho))$, where $\rho$ is an unknown quantum state and $E_i$ is a POVM operator sampled according to some distribution $D$, to output a quantum state $\sigma$ such that $\Pr_{E\la \cD}[\Tr(E\rho)-\Tr(E\sigma) \leq \epsilon]\geq 1-\gamma$ for small $\epsilon$ and $\gamma$.
Bshouty and Jackson~\cite{COLT:BJ95} consider another quantum generalization of the PAC learning model.
Roughly speaking, the model is the same as standard Valiant's PAC model except that the learner can receive $\sum_{x}\sqrt{D(x)}\ket{x}\ket{c(x)} $ instead of $(x,c(x))$, where $x$ is sampled according to the distribution $D$.
The subsequent works~\cite{QIP:AR07,CCC:AR17,PRA:GKZ19,arXiv:SSG23} study the learning model.
Aaronson introduces another learning model called shadow tomography~\cite{STOC:Aaronson18}.
In the model, a learner is given many copies of $\rho$, and POVM operators $E_1,\cdots,E_m$.
The learner needs to estimate $\Tr(\rho E_1),\cdots \Tr(\rho E_m)$ up to additive error $\epsilon$.
The subsequent works \cite{NP:HKP20,STOC:CR21} also study the model.

\subsection{Future Directions}

We raise several open questions that remain to be explored.
Below, we highlight the following two questions.

\paragraph{Show the $\mathbf{PP}$ hardness of properly learning pure states.}

We show that $\mathbf{PP}$ oracle helps to properly learn pure states generated by quantum polynomial time algorithms.
We leave the reverse direction as an open problem.
In other words, can we decide any language in $\mathbf{PP}$ efficiently assuming that we can properly learn arbitrary pure states generated by quantum polynomial-time algorithms?

Remark that in the classical learning theory, the hardness of proper PAC learning is characterized by $\mathbf{NP}$.
It is shown that PAC learning can be reduced to Occam learning~\cite{IPL:BEHW87,TCS:BP92,ML:Schapire90}, which can be formulated as a search problem in $\mathbf{NP}$.
Hence, $\mathbf{NP}$ oracle is sufficient to achieve proper PAC learning for arbitrary classical polynomial time algorithms.
On the other hand, it is shown that if we can properly learn arbitrary classical polynomial-time algorithms in the proper PAC model~\cite{ACM:PV88}, then we can decide any language in $\mathbf{NP}$.
Does a similar relationship hold between $\mathbf{PP}$ and proper pure state learning as the relationship between $\mathbf{NP}$ and proper PAC learning?

\paragraph{Characterize the hardness of learning mixed states.}

In this work, we observe that the existence of EFI implies that there exist mixed states generated by quantum polynomial-time algorithms, which is hard to properly learn.
This implies that it is unlikely to reduce \emph{proper} mixed-state learning to classical language because it is conjectured that any classical oracles do not help to break the security of EFI~\cite{arXiv:LMW23}.
Note that the argument cannot be applied to improper mixed-state learning, and
there exists the possibility that we can reduce improper mixed-state learning to classical language. 
We leave as an open problem whether we can reduce improper mixed-state learning to classical language or not.
We also leave as an open problem whether we can characterize the hardness of proper mixed-state learning in terms of unitary complexities~\cite{ITCS:GH22,FOCS:MY23,arXiv:JYTALH23} instead of classical languages.

\section{Preliminaries}
\subsection{Notations}
Here, we introduce basic notations we will use in this paper.
$x\la X$ denotes selecting an element $x$ from a finite set $X$ uniformly at random, and $y \la \cA(x)$ denotes assigning to $y$ the output of a quantum or probabilistic or deterministic algorithm $\cA$ on an input $x$.
Let $[n]\seteq \{1,\cdots,n\}$.
QPT stands for quantum polynomial time. 
A function $f : \N \ra \R$ is a negligible function if, for any constant $c$, there exists $\secp_0 \in \N$ such that for any $\secp>\secp_0$, $f(\secp) < 1/\secp^c$.
We write $f(\secp) \leq \negl(\secp)$ to denote $f(\secp)$ being a negligible function.
When we refer to polynomials, we mean functions $p:\N\ra \N$, where there exists a constant $c$ such that for all $n\in\N$, it holds that $p(n)\leq cn^c$.

For simplicity, we often write $\ket{0}$ to mean $\ket{0\cdots 0}$.
For any two quantum states $\rho_1$ and $\rho_2$, $F(\rho_1,\rho_2)$ is the fidelity between them, and $\mathsf{TD}(\rho_1,\rho_2)$ is the trace distance between them.

\paragraph{Sampleable Distribution.}
We say that a family of distributions $\cD=\{\cD_n\}_{n\in\N}$ is efficiently sampleable distribution if there exists a uniform QPT algorithm $D$ such that, for each $n\in\N$, the distribution of $D(1^n)$ is statistically identical to $\cD_n$. 

\begin{definition}[Kullback–Leibler Divergence]
    Let $P$ and $Q$ be a distribution over $\cZ_n=\bit^n$.
    \begin{align}
        D_{KL}(P||Q)\seteq\mathbb{E}_{z\la P}\left[\log(\frac{\Pr[z\la P]}{\Pr[z\la Q]})\right] =\sum_{z\in\cZ_n}\Pr[z\la P]\log(\frac{\Pr[z\la P]}{\Pr[z\la Q]}).
    \end{align}
\end{definition}

\subsection{Quantum Cryptographic Primitives}
In this section, we introduce quantum cryptographic primitives, which we will use in this work.
Note that throughout this work, as adversaries, we consider uniform QPT algorithms instead of non-uniform QPT algorithms.

\paragraph{One-Way State Generators (OWSGs)}
\begin{definition}[OWSGs with Pure State Outputs~\cite{Crypto:MY22}]\label{def:owsg_pure}
    A one-way state generator with pure state outputs is a set of uniform QPT algorithms $\Sigma\seteq(\keygen,\StateGen)$ such that:
\begin{itemize}
    \item[$\keygen(1^\secp)$:]
    It takes a security parameter $1^\secp$, and outputs a classical string $k$.
    \item[$\StateGen(1^\secp,k)$:]
    It takes a security parameter $1^\secp$ and $k$, and outputs a pure quantum state $\ket{\psi_k}$.
\end{itemize}
\paragraph{Security.}
We consider the following security experiment $\mathsf{Exp}_{\Sigma,\cA}(1^\secp)$.
\begin{enumerate}
    \item The challenger samples $k\la\keygen(1^\secp)$.
    \item $\cA$ can receive arbitrary polynomially many copies of $\ket{\psi_k}$.
    \item $\cA$ outputs $k^*$.
    \item The challenger measures $\ket{\psi_k}$ with $\{\ket{\psi_{k^*}}\bra{\psi_{k^*}}, I-\ket{\psi_{k^*}}\bra{\psi_{k^*}}\}$.
    \item The experiment outputs $1$ if the result is $\ket{\psi_{k^*}}\bra{\psi_{k^*}}$, and $0$ otherwise.
\end{enumerate}
We say that an OWSG with pure state outputs satisfies security if, for any uniform QPT adversaries $\cA$,
\begin{align}
    \Pr[\mathsf{Exp}_{\Sigma,\cA}(1^\secp)=1]\leq \negl(\secp).
\end{align}
\end{definition}
\begin{remark}
    We sometimes omit the security parameter $1^\secp$ of $\StateGen(1^\secp,k)$ when it is clear from the context.
\end{remark}

\begin{definition}[Secretly-Verifiable Statistically-Invertible OWSGs (SV-SI-OWSG)~\cite{Eprint:MY22}]
        A SV-SI-OWSG is a set of uniform QPT algorithms $\Sigma\seteq (\keygen,\StateGen)$ such that:
    \begin{description}
        \item[$\keygen(1^\secp)$:] It takes a security parameter $1^\secp$, and outputs a classical string $k$.
        \item[$\StateGen(1^\secp,k)$:] It takes a security parameter $1^\secp$ and $k$, and outputs $\psi_{\secp,k}$. 
    \end{description}
    \paragraph{Statistically Invertibility:}
    For any $(k,k^*)$ with $k\neq k^*$,
    \begin{align}
        \mathsf{TD}(\psi_{\secp,k},\psi_{\secp,k^*}) \geq 1-\negl(\secp).
    \end{align}
    \paragraph{Computational Non-Invertibility:}
    For any uniform QPT adversaries $\cA$ and any polynomials $t$,
    \begin{align}
        \Pr[k\la\cA(\psi_{\secp,k}^{\otimes t(\secp)}):k\la\keygen(1^\secp),\psi_{\secp,k}\la\StateGen(1^\secp,k)]\leq \negl(\secp).
    \end{align}
\end{definition}
Note that \cite{Eprint:MY22} shows that the existence of EFI is equivalent to that of SV-SI-OWSG.

\section{Model of Quantum State Learning}\label{sec:learning_model}

\begin{definition}[Quantum State Learnability]\label{def:state_learnability}
    Let $\cC$ be a quantum algorithm that takes $1^n$ with $n\in\N$, $x\in\cX_n=\bit^{\poly(n)}$, and outputs a $\poly(n)$-qubit quantum state $\psi_x$.
    We say that an algorithm $\cA$ can learn $\cC$ 
    if, for all polynomials $\epsilon$ and $\delta$, there exists a polynomial $t$ such that for all $x\in\cX_n$, we have 
    \begin{align}
        \Pr[
        \begin{array}{ll}
        \mathsf{TD}(\psi_{x},\psi_{y})\leq 1/\epsilon(n)
        \end{array}
        :
    \begin{array}{ll}
       \psi_{x}\la \cC(x,1^n)\\
        y\la \cA(\cC,\psi_{x}^{\otimes t(n)},1^n)\\
        \psi_y \la  \cC(y,1^n)
    \end{array}
        ]\geq 1-1/\delta(n)
    \end{align}
    for all sufficiently large $n\in\N$.
\end{definition}
\begin{remark}
    Throughout this work, we consider proper quantum state learning, where the learner's output is promised within $\cX_n$.
    The other variant of quantum state learning model is an improper one, where the learner's task is not to output $y$ such that $\cC(y)$ approximates $\cC(x)$, but output arbitrary polynomial-size quantum circuits $h$, whose output approximates $\cC(x)$.
    Note that if we can properly learn $\cC$, then we can improperly learn $\cC$.
    On the other hand, the other direction does not follow in general.
\end{remark}

In this work, we also consider the average-case hardness of quantum state learning, which is defined as follows:
\begin{definition}[Average-Case Hardness of Quantum State Learning]\label{def:AH_QLearn}
    Let $\cC$ be a quantum algorithm that takes $1^n$ with $n\in\N$, $x\in\cX=\bit^{\poly(n)}$, and outputs a $\poly(n)$-qubit quantum state $\psi_x$.
    We say that $\cC$ is hard on average if there exists an efficiently sampleable distribution $\cS_n$ over $\cX_n$ and polynomials $\epsilon$ and $\delta$ such that for all polynomials $t$ and all QPT algorithms $\cA$, we have
    \begin{align}
        \Pr[
        \begin{array}{ll}
        \mathsf{TD}(\psi_x,\psi_y)\leq 1/\epsilon(n)
        \end{array}
        :
    \begin{array}{ll}
    x\la \cS_n \\
    \psi_x\la \cC(x,1^n)\\
    y\la \cA(\cC,\psi_{x}^{\otimes t(n)},1^n)\\
    \psi_y\la \cC(y,1^n)
    \end{array}
        ]\leq 1/\delta(n)
    \end{align}
    for all sufficiently large $n\in\N$.
\end{definition}

\section{Quantum State Learning Algorithm with Classical Oracle}\label{sec:quantum_state_learning}

In this section, we prove the following main theorem.
\begin{theorem}[Quantum State Learning Algorithm with $\mathbf{PP}$ oracle]\label{thm:Learning_w_Oracle}
    For all quantum polynomial-time algorithms $\cC$ with pure state outputs,
    there exists a language $\cL\in\mathbf{PP}$ and an oracle-aided quantum polynomial-time algorithm $\cA^{\cL}$ that can learn $\cC$ in the sense of \cref{def:state_learnability}.
\end{theorem}

To prove \cref{thm:Learning_w_Oracle}, we use the following \cref{lem:Algo_for_SQMLH,lem:Hoeffding,lem:random_measurement}, and the following \cref{prop:frobenius}.

\begin{lemma}[Algorithm for Quantum Maximum Likelihood Problem]\label{lem:Algo_for_SQMLH}
    For all quantum polynomial-time algorithms $\cC$ that takes $1^n$ and $x\in\cX_n=\bit^{\poly(n)}$ as input and outputs a classical string $z\in\cZ_n=\bit^{\poly(n)}$,
    there exists a language $\cL\in \mathbf{PP}$ and an oracle-aided PPT algorithm $\cA^\cL$ such that for all $z\in\cZ_n$ with $0<\max_{x\in\cX_n}\Pr[z\la\cC(1^n,x)]$ and all polynomial $\epsilon$ and $\delta$,
    \begin{align}
         \Pr_{h\la\cA^{\cL}(1^n,\cC,z)}\left[\frac{\max_{x\in\cX_n}\Pr[z\la \cC(x,1^n) ]}{\left(1+1/\epsilon(n)\right)}\leq   \Pr[z\la \cC(h,1^n)] \leq \max_{x\in\cX_n}\Pr[z\la \cC(x,1^n) ]\right]\geq 1-2^{-\delta(n)}
    \end{align}
    for all sufficiently large $n\in\N$.
\end{lemma}
We give the proof of \cref{lem:Algo_for_SQMLH} in \cref{sec:Alog_for_SQMLH}.

\begin{lemma}[Random Measurement preserves Distance~\cite{CCC:AE07}]\label{lem:random_measurement}
There exists constants $A$ and $B$ such that, for any $n$-qubit mixed states $\rho_1,\rho_2$ with $2^{-n/3} <A\norm{\rho_1-\rho_2}_F^4$, we have
\begin{align}
     B\norm{\rho_0-\rho_1}_F\leq\norm{\cM(\rho_1)-\cM(\rho_2)}_{1},
\end{align}
where $\cM$ is a POVM with respect to an $2^{-n/3}$-approximate (4, 4)-design.
\end{lemma}

We refer the proof of \cref{lem:random_measurement} to \cite{CCC:AE07}.

\begin{lemma}[Followed by Hoeffding Inequality]\label{lem:Hoeffding}
    Let $n\in\N$, $\cZ=\bit^{\poly(n)}$, and let $\cX$ be a family of function such that $X:\cZ\ra [0,M]$ for all $X\in\cX$, where
    $M\in\R$.
    Let $\cC$ be an arbitrary distribution over $\cZ$.
    If 
    \begin{align}
        T\geq \frac{M^2}{\epsilon^2}\abs{\log{\abs{\cX}}+\log(\delta)},
    \end{align}
    then, for all $X\in\cX$, we have
    \begin{align}
        \abs{\frac{1}{m}\sum_{i\in[T]}X(z_i)-  \mathbb{E}_{z\la\cC}[X(z)] }\leq \epsilon
    \end{align}
    with probability greater than $1-1/\delta$, where the probability is taken over $\{z_i\}_{i\in[T]}\la \cC^T$.
\end{lemma}

\begin{proposition}\label{prop:frobenius}
When $\rho_0$ and $\rho_1$ are pure states,
\begin{align}
    \norm{\rho_0-\rho_1}_1= \sqrt{2}\norm{\rho_0-\rho_1}_F 
\end{align}
\end{proposition}

We describe the proof of \cref{thm:Learning_w_Oracle}.

\begin{proof}[Proof of \cref{thm:Learning_w_Oracle}]          
    Let $\cL$ be a language in $\mathbf{PP}$ and $\cA^{\cL}$ given in \cref{lem:Algo_for_SQMLH}.
    Let us describe our notations, and our oracle-aided learning algorithm $\cB^{\cL}$.
    \begin{description}
        \item[Notations:] $ $
        \begin{itemize}
            \item Let $\cC$ be a quantum polynomial-time algorithm that takes as input $1^n$ and $x\in\cX_n=\bit^{\poly(n)}$, and outputs a $\ket{\psi_x}$.
            \item Let $\cC^*_{\alpha}$ be a quantum polynomial-time algorithm that takes as input $1^n$ and $x\in \cX_n$, generates $\ket{\psi_x}$ by running $\cC(1^n,x)$, measures $\ket{\psi_x}$ with $\cM$, and obtains the measurement outcome $z\in\cZ_n=\bit^{\ell(n)}$, where $\cM$ is a POVM measurement given in \cref{lem:random_measurement}.
            $\cC^*_{\alpha}$ outputs the measurement outcome $z$ with probability $1-\alpha$, and, with probability $\alpha$, outputs a uniformly random $Z\la \cZ_n$.
            \item Let $\cC^{*T}_{\alpha}$ be a $T$-repetition version of $\cC^*_{\alpha}$.
            More formally,
            it takes as input $1^n$ and $x\in \cX_n$, and generates $\{\ket{\psi_x}_i\}_{i\in[T]}$.
            Then, for each $i\in[T]$, it  measures $\ket{\psi_x}_i$ with $\cM$, obtains the measurement outcome $z_i$, and sets $Z_i=z_i$ with probability $1-\alpha$ and $Z_i\la\cZ_n$ with probability $\alpha$, where $\cM$ is a POVM measurement given in \cref{lem:random_measurement}.
            Finally, it outputs $\{Z_i\}_{i\in[T]}$.
            \item Let $\epsilon^*(n)=\frac{4\epsilon(n)^2}{B^{2}}$, where $B$ is a constant given in \cref{lem:random_measurement}.
            \item Let $T\geq \epsilon^{*2}(n)\left(\ell(n)+1\right)^2\left(\log(\abs{\cX_n})+\log(8\delta)\right)$ be a number of copies of $\ket{\psi}$, which $\cB$ uses.
        \end{itemize}
        \item[The description of $\cB^\cL$:]$ $
        \begin{enumerate}
            \item
            It takes as input $1^n$, a quantum polynomial-time algorithm $\cC$, and $\{\ket{\psi_x}_i\}_{i\in[T]}$, which is generated by $\cC(1^n,x)$ with unknown $x$.
            \item For all $i\in[T]$, it measures $\ket{\psi_x}_i$ with the POVM $\cM$ given in \cref{lem:random_measurement},  obtains a classical string $z_i$, and sets $Z_i=z_i$ with probability $1/2$ and $Z_i\la\cZ_n$ with probability $1/2$.
            \item For all $i\in[T]$, it runs $\cA^\cL(1^n,\cC^{*T}_{1/2},\{Z_i\}_{i\in[T]})$ given in \cref{lem:Algo_for_SQMLH}, and obtains $h\in\cX_n$ such that
            \begin{align}
                \frac{\max_{x\in\cX_n}\left\{\Pr[\{Z_i\}_{i\in[T]}\la \cC^{*T}_{1/2}(1^n,x)]\right\}}{\Pr[\{Z_i\}_{i\in[T]}\la\cC^{*T}_{1/2}(1^n,h) ]} \leq
                \left(1+1/p(n)\right)\leq 2^{\left(\frac{T}{2\epsilon^*(n)}\right)}
            \end{align}
            with some polynomial $p$.
        \end{enumerate}
    \end{description}

Now, we show that
\begin{align}
    \frac{1}{2}\norm{\ket{\psi_x}-\ket{\psi_h}}_1\leq 1/\epsilon(n)
\end{align}
with probability $1-2^{-\delta(n)}$, where the probability is taken over $\{Z_i\}_{i\in[T]}\la\cC_{1/2}^{*T}$ and the internal randomness of $\cA$.

First, for all $y\in\cX_n$, the output $h$ of $\cB^{\cL}$ satisfies
\begin{align}
    \mathbb{E}_{Z\la\cC^*_{1/2}(x)}\left[\log(\frac{\Pr[Z\la\cC^*_{1/2}(y)]}{\Pr[Z\la \cC_{1/2}^*(h)]})\right]\leq 1/\epsilon^*(n) 
\end{align}
with probability $1-2^{-\delta(n)}$,
where the probability is taken over $\{Z_i\}_{i\in[T]}\la\cC_{1/2}^{*T}$ and the internal randomness of $\cA$.
We will prove the inequality above later.
Once we have obtained the inequality above, the remaining part straightforwardly follows.
First, we have
\begin{align}
   & D_{KL}(\cC_{1/2}^*(x)||\cC_{1/2}^*(h))\\
   &\seteq\mathbb{E}_{Z\la\cC^*_{1/2}(x)}\left[\log(\frac{\Pr[Z\la\cC^*_{1/2}(x)]}{\Pr[Z\la \cC_{1/2}^*(h)]})\right]\leq 1/\epsilon^*(n)
\end{align}
with probability $1-2^{-\delta}$.
Furthermore, from Pinsker's inequality, we have
\begin{align}
    \norm{\cC_{1/2}^*(x)-\cC_{1/2}^*(h)}_1\leq \sqrt{2 D_{KL}(\cC_{1/2}^*(x)||\cC_{1/2}^*(h))}\leq \sqrt{\frac{2}{\epsilon^*(n)}}.
\end{align}
On the other hand, for all $0\leq\alpha\leq1$, we have
\begin{align}
    &\norm{\cC_\alpha^*(x)-\cC_\alpha^*(h)}_1\\
    &\seteq\frac{1}{2}\sum_{Z}\abs{\Pr[Z\la \cC_\alpha^*(x)]-\Pr[Z\la\cC_\alpha^*(y)]}\\ 
    &=\frac{1}{2}\sum_{Z}\abs{(1-\alpha)\Pr[Z\la \cC_0^*(x)]-(1-\alpha)\Pr[Z\la\cC_0^*(y)]}\\
    &=(1-\alpha)\cdot \frac{1}{2}\sum_{Z}\abs{\Pr[Z\la \cC_0^*(x)]-\Pr[Z\la\cC_0^*(y)]}\\
    &=(1-\alpha)\norm{\cC_0^*(x)-\cC_0^*(h)}_1.
\end{align}
Hence,
\begin{align}
    \norm{\cC_0^*(x)-\cC_0^*(h)}_1\leq 2\norm{\cC_{1/2}^*(x)-\cC_{1/2}^*(h)}_1\leq 2\sqrt{\frac{2}{\epsilon^*(n)}}.
\end{align}
Suppose that 
\begin{align}
    \norm{\ket{\psi_x}-\ket{\psi_h}}_1\leq \left(\frac{4\cdot 2^{-n/3}}{A}\right)^{1/4},
\end{align}
where $A$ is a constant given in \cref{lem:random_measurement}.
Then, for all polynomial $\epsilon$, we have
\begin{align}
\norm{\ket{\psi_x}-\ket{\psi_h}}_1\leq 1/\epsilon(n)    
\end{align}
for all sufficiently large $n\in\N$.
On the other hand, if
\begin{align}
    \left(\frac{4\cdot 2^{-n/3}}{A}\right)^{1/4}<\norm{\ket{\psi_x}-\ket{\psi_h}}_1,
\end{align}
then \cref{lem:random_measurement,prop:frobenius} imply that 
\begin{align}
    \frac{1}{2}\norm{\ket{\psi_x})-\ket{\psi_h}}_1=
    \frac{1}{\sqrt{2}}\norm{\ket{\psi_x}-\ket{\psi_h}}_F
    \leq
    \frac{1}{B\sqrt{2}}\norm{\cC^*_0(x)-\cC^*_0(h)}_1\leq \frac{2}{B\sqrt{\epsilon^*(n)}}\seteq\frac{1}{\epsilon(n)}.
\end{align}

Now, we prove the remaining part, that is, for all $y\in\cX_n$,
\begin{align}
    \mathbb{E}_{Z\la\cC^*_{1/2}(x)}\left[\log(\frac{\Pr[Z\la\cC^*_{1/2}(y)]}{\Pr[Z\la \cC_{1/2}^*(h)]})\right]\leq 1/\epsilon^*(n)
\end{align}
with $1-2^{-\delta(n)}$, where the probability is taken over $\{Z_i\}_{i\in[T]}\la \cC_{1/2}^{*T}$ and the internal randomness of $\cA$.
From the construction of $\cB^{\cL}$, for all $\{Z_i\}_{i\in[T]}$, $\cB$ outputs $h$ such that 
\begin{align}
    &\frac{\max_{s\in\cX_n}\left\{\Pr[\{Z_i\}_{i\in[T]}\la \cC^{*T}_{1/2}(s) ]\right\}}{\Pr[\{Z_i\}_{i\in[T]}\la\cC^{*T}_{1/2}(h) ]}\\
    &=\frac{\max_{s\in\cX_n}\left\{\Pi_{i\in[T]}\Pr[Z_i\la \cC^{*}_{1/2}(s) ]\right\}}{\Pi_{i\in[T]}\Pr[Z_i\la\cC^{*}_{1/2}(h) ]} \leq
    \left(1+1/p(n)\right)\leq 2^{\left(\frac{T}{2\epsilon^*(n)}\right)}
\end{align}
with probability $1-2^{-\delta(n)}/2$, where the probability is taken over the internal randomness of $\cA$.
By taking logarithmic and dividing with $T$, we have
\begin{align}
    \left(\frac{1}{T}\sum_{i\in[T]}\log(\frac{1}{\Pr[Z_i\la \cC^*_{1/2}(h)]}) +\max_{s\in\cX_n}\frac{1}{T}\sum_{i\in[T]}\log(\Pr[Z_i\la \cC^*_{1/2}(s)])  \right) \leq \frac{1}{2\epsilon^*(n)}.
\end{align}
This implies that
for all $y\in\cX_n$, 
\begin{align}
&\left(\frac{1}{T}\sum_{i\in[T]}\log(\frac{1}{\Pr[Z_i\la \cC^*_{1/2}(h)]}) -\frac{1}{T}\sum_{i\in[T]}\log(\frac{1}{\Pr[Z_i\la \cC^*_{1/2}(y)]})  \right)\\
&=\left(\frac{1}{T}\sum_{i\in[T]}\log(\frac{1}{\Pr[Z_i\la \cC^*_{1/2}(h)]}) +\frac{1}{T}\sum_{i\in[T]}\log(\Pr[Z_i\la \cC^*_{1/2}(y)])  \right)\\
&\leq \left(\frac{1}{T}\sum_{i\in[T]}\log(\frac{1}{\Pr[Z_i\la \cC^*_{1/2}(h)]}) +\max_{s\in\cX_n}\frac{1}{T}\sum_{i\in[T]}\log(\Pr[Z_i\la \cC^*_{1/2}(s)])  \right) 
\leq \frac{1}{2\epsilon^*(n)}
\end{align}
with probability $1-2^{-\delta(n)}/2$, where the probability is taken over the internal randomness $\cA$.
Furthermore, because
\begin{align}
    0\leq\log(\frac{1}{\Pr[Z\la\cC^*_{1/2}(x)]})\leq \ell(n)+1
\end{align}
for all $Z\in\cZ_n=\bit^{\ell(n)}$ and $x\in\cX_n$, we have
\begin{align}
    \left(\mathbb{E}_{Z\la \cC_{1/2}^*(x)}\left[\log(\frac{1}{\Pr[Z\la \cC^*_{1/2}(h)]})\right]-\frac{1}{T}\sum_{i\in[T]}\log(\frac{1}{\Pr[Z_i\la \cC^*_{1/2}(h)]})\right)\leq \frac{1}{4\epsilon^*(n)}
\end{align}
and 
\begin{align}
    \left(\frac{1}{T}\sum_{i\in[T]}\log(\frac{1}{\Pr[Z_i\la \cC^*_{1/2}(y)]})-\mathbb{E}_{Z\la \cC_{1/2}^*(x)}\left[\log(\frac{1}{\Pr[Z\la \cC^*_{1/2}(y)]})\right]\right)\leq \frac{1}{4\epsilon^*(n)}
\end{align}
with probability $1-2^{-\delta(n)}/2$, where the probability is taken over $\{Z_i\}_{i\in[T]}\la \cC_{1/2}^{*T}(x)$.
Note that, here, we use \cref{lem:Hoeffding}.

By summing up three inequalities above, we have
\begin{align} 
      &\left(\mathbb{E}_{Z\la \cC_{1/2}^*(x)}\left[\log(\frac{1}{\Pr[Z\la \cC^*_{1/2}(h)]})\right]-\mathbb{E}_{Z\la \cC_{1/2}^*(x)}\left[\log(\frac{1}{\Pr[Z\la \cC^*_{1/2}(y)]})\right]\right)\leq 1/\epsilon^*(n)
\end{align}
with probability $1-2^{-\delta(n)}$, where the probability is taken over $\{Z_i\}_{i\in[T]}\la \cC_{1/2}^{*T}(x)$ and the internal randomness of $\cA $.
This implies that
\begin{align}
    \mathbb{E}_{Z\la\cC^*_{1/2}(x)}\left[\log(\frac{\Pr[Z\la\cC^*_{1/2}(y)]}{\Pr[Z\la \cC_{1/2}^*(h)]})\right]\leq 1/\epsilon^*(n)
\end{align}
with probability $1-2^{-\delta(n)}$, where the probability is taken over $\{Z_i\}_{i\in[T]}\la \cC_{1/2}^{*T}(x)$ and the internal randomness of $\cA$.
\end{proof}

\color{black}

\subsection{Algorithm for Quantum Maximum-Likelihood Problem}\label{sec:Alog_for_SQMLH}

In this section, we prove \cref{lem:Algo_for_SQMLH}.
To prove \cref{lem:Algo_for_SQMLH}, we use the following \cref{lem:Extrapolation_PP}.
\begin{lemma}[\cite{FR99}]\label{lem:Extrapolation_PP}
    There exists a language $\cL\in\mathbf{PP}$ and an oracle-aided deterministic polynomial-time algorithm $\cB^{\cL}$ such that the following holds.
    \begin{align}
        \cB^{\cL}(1^n,\cC,z)=\Pr[z\la\cC(1^n)].
    \end{align}
    Here, 
    $\cC$ is a quantum polynomial-time algorithm that takes $1^n$ as input, and outputs $z$.
\end{lemma}
We refer the proof to \cite{FR99}.
Now, we prove \cref{lem:Algo_for_SQMLH}.
\begin{proof}[Proof of \cref{lem:Algo_for_SQMLH}]
    Let $\cL\in\mathbf{PP}$ and $\cB$ be a deterministic polynomial time algorithm given in \cref{lem:Extrapolation_PP}.
    We describe our notations and our PPT algorithm $\cA^\cL$.
    \begin{description}
        \item[Notations:]$ $
        \begin{itemize}
            \item Let $\cC$ be a QPT algorithm that takes as input $1^n$ and $x\in\cX_n$, and outputs $z\in\cZ_n$.
            Let $\ell$ be the length of $x\in\cX_n$.
            \item Without loss of generality, we can assume that the QPT algorithm $\cC(1^n,x)$ works as follows:
            \begin{enumerate}
                \item It prepares a quantum state
                \begin{align}
                    U_{\cC(x)}\ket{0}_{YV}\seteq \ket{\cC(x)}_{YV}=\sum_{z\in\cZ}\sqrt{\Pr[z\la\cC(x)]}\ket{z}_Y\ket{\phi_{x,z}}_V.
                \end{align}
                \item It measures the $Y$ register, and outputs the measurement outcome.
            \end{enumerate}
            \item Let $T\geq \frac{ \delta+\log(\abs{\cX_n})}{\log(1+1/\epsilon)}$ so that 
            \begin{align}
                \abs{\cX_n}\left(\frac{1}{1+1/\epsilon} \right)^T\leq 2^{-\delta}.
            \end{align}
            \item Let $\widehat{\cQ}[z]$ be a quantum algorithm with post-selection working as follows: 
            \begin{enumerate}
                \item It prepares 
                \begin{align}
                    U_\cQ\ket{0}_{WY_1V_1\cdots Y_TV_T}&\seteq \frac{1}{\sqrt{\abs{\cX_n}}}\sum_{x\in\cX_n}\ket{x}_W\bigotimes_{i\in[T]}\left( \ket{\cC(x)}_{Y_iV_i}\right)\\
                    &=
                    \frac{1}{\sqrt{\abs{\cX_n}}}\sum_{x\in\cX_n}\ket{x}_W\bigotimes_{i\in[T]} \left(\sum_{z\in\cZ_n}\sqrt{\Pr[z\la\cC(x)]}\ket{z}_{Y_i}\ket{\phi_{x,z}}_{V_i}\right),
                \end{align}
                measures the $Y_1\cdots Y_T$ in the computational basis, and obtains $\{y_i\}_{i\in[T]}$.
                Repeat this step until $y_i=z$ for all $i\in[T]$.
                Let $\ket{\psi[z]}_{WY_1V_1\cdots Y_TV_T}$ be the post-selected quantum state.
                Note that we have
                \begin{align}
                    \ket{\psi[z]}_{WY_1V_1\cdots Y_TV_T}=\frac{1}{\sqrt{\sum_{x\in\cX_n}\Pr[z\la\cC(x)]^T}}\sum_{h\in\cX_n}
                    \sqrt{\Pr[z\la\cC(h)]^T}\ket{h}_W\bigotimes_{i\in[T]}\left(\ket{z}_{Y_i}\ket{\phi_{h,z}}_{V_i}\right).
                \end{align}
                \item Measure the $W$ register and output the measurement outcome $h$.
            \end{enumerate}
        \end{itemize}
        Our oracle aided PPT algorithm $\cA^{\cL}$ simulates $\widehat{\cQ}[z]$ by working as follows.
        \item[Description of $\cA^{\cL}$:] $ $
        \begin{itemize}
            \item As an input of $\cA^{\cL}$, it receives $1^n$, a QPT algorithm $\cC$, and a classical string $z\in\cZ_n$.
            \item For $i\in[\ell]$, $\cA^{\cL}$ sequentially works as follows:
            \begin{itemize}
            \item For each $b\in\bit$, it computes
            \begin{align}
               \Pr[b | x_{i-1},\cdots,x_1,z]= \frac{\Pr[b,x_{i-1},\cdots, x_1,z]}{\Pr[x_{i-1},\cdots, x_1,z ]}.
            \end{align}
            Here, 
            \begin{align}
                \Pr[x_{i-1},\cdots,x_1,z]=\Tr(\left(\ket{x_{i-1},\cdots,x_1}\bra{x_{i-1},\cdots,x_1}_{W[i-1]}\bigotimes_{i\in[T]} \ket{z}\bra{z}_{Y_i}\right) U_\cQ\ket{0}_{WY_1V_1\cdots Y_TV_T} ),
            \end{align}
            where $W[i-1]$ is the register in the first $(i-1)$-bits of $W$.
            For this procedure, we use the algorithm $\cB$ and $\cL\in\mathbf{PP}$ given in \cref{lem:Extrapolation_PP}.
            \item Set $x_i=b$ with probability 
            \begin{align}
                \Pr[b | x_{i-1},\cdots,x_1,z]= \frac{\Pr[b,x_{i-1},\cdots, x_1,z]}{\Pr[x_{i-1},\cdots, x_1,z ]}.
            \end{align}
            \end{itemize}
            \item Output $h\seteq x_\ell,\cdots,x_1$.
        \end{itemize}
    \end{description}
From the construction of $\cA^{\cL}$, 
we have
\begin{align}
    \Pr[x\la\cA^{\cL}(1^n,\cC,z)]=\frac{\Pr[x_\ell,\cdots,x_1,z]}{\Pr[z]}
\end{align}
for all $x\in\cX_n$.
From the definition, $\frac{\Pr[x_\ell,\cdots,x_1,z]}{\Pr[z]}$ is exactly equal to the probability that $\widehat{Q}[z]$ outputs $x$.
This means that the output distribution of $\cA^{\cL}(1^n,\cC,z)$ is equivalent to that of $\widehat{Q}[z]$.
Therefore, it is sufficient to show that
\begin{align}
\Pr_{h\la \widehat{\cQ}[z](1^n)}\left[ \frac{\max_{x\in\cX_n}\Pr[z\la\cC(x)] }{1+1/\epsilon}\leq \Pr[z\la\cC(h)] \right]   \geq 1-2^{-\delta}.
\end{align}
Let $\mathsf{Bad}_z$ be a subset of $ \cX_n$ such that
\begin{align}
\mathsf{Bad}_z\seteq \left\{h\in\cX_n:\Pr[z\la\cC(h)]<\frac{\max_{x\in\cX_n}\Pr[z\la\cC(x)] }{1+1/\epsilon}\right \}.
\end{align}
From the construction of $\widehat{\cQ}[z]$, we have
\begin{align}
    \Pr[h\la \widehat{\cQ}[z]]= \frac{\Pr[z\la \cC(h)]^T}{ \sum_{x\in\cX_n} \Pr[z\la\cC(x)]^T }
\end{align}
for all $h\in\cX_n$.
Then, for all $h\in\mathsf{Bad}_z$, we have
\begin{align}
    \Pr[h\la \widehat{\cQ}[z]]= \frac{\Pr[z\la \cC(h)]^T}{ \sum_{x\in\cX_n} \Pr[z\la\cC(x)]^T }\leq 
    \frac{\Pr[z\la \cC(h)]^T}{ \max_{x\in\cX_n}\Pr[z\la\cC(x)]^T }<\left(\frac{1}{1+1/\epsilon}\right)^T.
\end{align}
Therefore, we have
\begin{align}
    \sum_{h\in\mathsf{Bad}_z} \Pr[h\la \widehat{\cQ}[z]]<\sum_{h\in\mathsf{Bad}_z}\left(\frac{1}{1+1/\epsilon}\right)^T=\abs{\mathsf{Bad}_z}\left(\frac{1}{1+1/\epsilon}\right)^T\leq \abs{\cX_n}\left(\frac{1}{1+1/\epsilon}\right)^T\leq 2^{-\delta},
\end{align}
which completes the proof.

\end{proof}

\section{Hardness of Quantum State Learning}\label{sec:hardness}
\subsection{Equivalence between Average-Case Hardness of Learning Pure States and OWSGs}
\begin{theorem}\label{thm:owsg_a_h_learn}
    The existence of OWSGs with pure state outputs is equivalent to that of a polynomial-size quantum circuit with pure state outputs, which is hard on average in the sense of \cref{def:AH_QLearn}.
\end{theorem}

\begin{remark}
This result works only for "proper" quantum state learning, and it is unclear how to show the hardness of improper quantum state learning from OWSGs as far as we understand.
On the other hand, we can show the hardness of improper quantum state learning from the existence of pseudorandom state generators.
\end{remark}

\cref{thm:owsg_a_h_learn} relatively straightforwardly follows from the definition.
We describe the proof for clarity.
\begin{proof}[Proof of \cref{thm:owsg_a_h_learn}]

First, we show that if there exist OWSGs with pure state outputs, then there exists a quantum polynomial-time algorithm $\cC$ with pure state outputs, which is hard on average.

\paragraph{OWSGs $\ra$ Average-Case Hardness of Quantum State Learning.}
It is sufficient to construct a quantum polynomial-time algorithm $\cC$ such that there exists a sampleable distribution $\cS_n$ over $\cX_n$ and polynomials $\epsilon$ and $\delta$ such that for all polynomials $t$,
\begin{align}
        \Pr[
        \begin{array}{ll}
        F(\ket{\psi_{x}},\ket{\psi_{x^*}})\geq 1-1/\epsilon(n)
        \end{array}
        :
    \begin{array}{ll}
    x\la \cS_n \\
       \ket{\psi_{x}}\la \cC(1^n,x)\\
        x^*\la \cA(\ket{\psi_{x}}^{\otimes t(n)},1^n)\\
        \ket{\psi_{x^*}}\la \cC(1^n,x^*)
    \end{array}
        ]\leq 1-1/\delta(n)
    \end{align}
    for all sufficiently large $n\in\N$.

We define $\cC$ and $\cS$ as follows:
\begin{description}
    \item[The description of $\cS_{n}$:]$ $
    \begin{itemize}
        \item $\cS_{n}$ is the same as the distribution of $\keygen(1^n)$.
    \end{itemize}
    \item[The description of $\cC$:]$ $
    \begin{itemize}
        \item $\cC$ runs $\StateGen(1^n, x)$, and outputs its output.
    \end{itemize}
\end{description}
For contradiction, assume that for all polynomials $\epsilon$ and $\delta$ there exists a polynomial $t$, and a QPT learner $\cB$ such that
\begin{align}
    \Pr[
    \begin{array}{ll}
    F(\ket{\psi_{x}},\ket{\psi_{x^*}})\geq 1-1/\epsilon(n)
     \end{array}
        :
    \begin{array}{ll}
       x\la \cS_n  \\
       \ket{\psi_{x}}\la\cC(1^n,x) \\
        x^*\la \cB(\ket{\psi_{x}}^{\otimes t(n)},1^n)\\
        \ket{\psi_{x^*}}\la\cC(1^n,x^*)
    \end{array}
        ]> 1-1/\delta(n)
\end{align}
for infinitely many $n\in\N$.
Now, by using $\cB$, we construct a QPT adversary $\cA$ that breaks the security of OWSGs $\Sigma$.
We describe the $\cA$ in the following.
\begin{enumerate}
    \item $\cA$ receives sufficiently many copies of $\ket{\psi_x}$, where $x\la\keygen(1^n)$ and $\ket{\psi_x}\la\StateGen(1^n,x)$.
    \item $\cA$ passes them to $\cB$, and receives $x^*$.
    \item $\cA$ sends $x^*$ to the challenger.
    \item The challenger measures $\ket{\psi_x}$ with $\{\ket{\psi_{x^*}}\bra{\psi_{x^*}},I-\ket{\psi_{x^*}}\bra{\psi_{x^*}}\}$.
    \item The challenger outputs $1$ if the measurement result is $\ket{\psi_{x^*}}\bra{\psi_{x^*}}$ is obtained, and $0$ otherwise.
\end{enumerate}
Let us define
\begin{align}
    \mathsf{Good}\seteq \left\{(x,x^*):\abs{\bra{\psi_{x}}\ket{\psi_{x^*}}}^2\geq 1-1/\epsilon(n)\right\}.
\end{align}
Now, we compute the probability that the challenger outputs $1$.
\begin{align}            
&\Pr[1\la\mbox{Challenger}:x\la\keygen(1^n),\ket{\psi_x}\la\StateGen(1^\secp,x), x^*\la\cA(\ket{\psi_x}^{\otimes t(n)})]\\
&=\Pr[1\la\mbox{Challenger}
        :
    \begin{array}{ll}
       x\la \cS_n  \\
       \ket{\psi_x}\la \cC(1^n,x)\\
        x^*\la \cB(\ket{\psi_{x}}^{\otimes t(\secp)},1^n)\\
        \ket{\psi_{x^*}}\la \cC(1^n,x^*)
    \end{array}
        ]\\
&=\sum_{x,x^*}\Pr[x\la\cS_n]\Pr[x^*\la \cB(\ket{\psi_{x}}^{\otimes t(n)})]\abs{\bra{\psi_{x^*}}\ket{\psi_{x}}}^2\\
&\geq\sum_{x,x^*\in \mathsf{Good}}\Pr[x\la\cS_n]\Pr[x^*\la\cB(\ket{\psi_{x}}^{\otimes t(n)})]\abs{\bra{\psi_{x}}\ket{\psi_{x^*}}}^2\\
&\geq \sum_{x,x^*\in \mathsf{Good}}\Pr[x\la\cS_n]\Pr[x^*\la\cB(\ket{\psi_{x}}^{\otimes t(n)})](1-1/\epsilon(n))\\
&\geq (1-1/\delta(n))(1-1/\epsilon(n))
\end{align}
for infinitely many $n\in\N$.
This contradicts that $\Sigma$ is OWSG with pure state outputs, which completes the proof.
\paragraph{Average-case hardness of quantum state learning $\ra$ OWSGs }
Let $\cC$ be a quantum polynomial-time algorithm that takes $1^n$ and $x\in\cX_n=\bit^{\poly(n)}$ and outputs $\ket{\psi_x}$, and let $\cS_n$ be a sampleable distribution over $\cX_n$ and let $\epsilon$ and $\delta$ be polynomials such that for all polynomials $t$ and all QPT algorithms $\cA$, 
\begin{align}
    \Pr[
    \begin{array}{ll}
    F(\ket{\phi_{x^*}},\ket{\psi_{x}})\geq 1/\epsilon(n)
    \end{array}
        :
    \begin{array}{ll}
       x\la \cS_n  \\
       \ket{\psi_{x}}\la \cC(1^n,x)\\
        x^*\la \cA(\ket{\psi_{x}}^{\otimes t(n)},1^n)\\
        \ket{\psi_{x^*}}\la\cC(1^n,x^*)
    \end{array}
        ]\leq 1/\delta(n)
\end{align}
for infinitely many $n\in\N$.
Let $\Samp$ be an algorithm that takes $1^n$ and perfectly approximates $\cS_n$.

We describe the construction of OWSGs $(\keygen,\StateGen)$:
\begin{description}
    \item[$\keygen(1^n)$:] $ $
    \begin{itemize}
        \item $\keygen(1^n)$ runs $\Samp(1^n)$ and obtains $x\in\cX_n$.
    \end{itemize}
    \item[$\StateGen(1^n,x)$:]$ $
    \begin{itemize}
        \item $\StateGen(1^n,x)$ outputs the output of $\cC(1^n,x)$.
    \end{itemize}
\end{description}
For contradiction, assume that there exist polynomials $t$, and $p$ and a QPT algorithm $\cB$ such that
\begin{align}
   \frac{1}{\delta(n)}+\frac{1}{\epsilon(n)}-\frac{1}{\delta(n)\cdot\epsilon(n)}= 1/p(n)<\sum_{x,x^*}\Pr[x\la\cS_n]\Pr[x^*\la \cB(\ket{\psi_{x}}^{\otimes t(n)})]\abs{\bra{\psi_{x^*}}\ket{\psi_{x}}}^2
\end{align}
for infinitely many $n\in\N$.
For some polynomial, let us define
\begin{align}
    &\mathsf{Good}\seteq\left\{(x,x^*):\frac{1}{\epsilon(n)}<\abs{\bra{\psi_x}\ket{\psi_{x^*}}}^2\leq 1\right\}\\
    &\mathsf{Bad}\seteq\left\{(x,x^*):\abs{\bra{\psi_x}\ket{\psi_{x^*}}}^2\leq \frac{1}{\epsilon(n)}\right\}
\end{align}
Now, we have
\begin{align}
    1/p(n)\leq &\sum_{x,x^*}\Pr[x\la\cS_n]\Pr[x^*\la \cB(\ket{\psi_{x}}^{\otimes t(n)})]\abs{\bra{\psi_{x^*}}\ket{\psi_{x}}}^2\\
    &=\sum_{x,x^*\in\mathsf{Good}}\Pr[x\la\cS_n]\Pr[x^*\la \cB(\ket{\psi_{x}}^{\otimes t(n)})]\abs{\bra{\psi_{x^*}}\ket{\psi_{x}}}^2
    +\sum_{x,x^*\in\mathsf{Bad}}\Pr[x\la\cS_n]\Pr[x^*\la \cB(\ket{\psi_{x}}^{\otimes t(n)})]\abs{\bra{\psi_{x^*}}\ket{\psi_{x}}}^2\\
    &\leq \sum_{x,x^*\in\mathsf{Good}}\Pr[x\la\cS_n]\Pr[x^*\la\cB(\ket{\psi_x}^{\otimes t(n)})]+\sum_{x,x^*\in\mathsf{Bad}}\Pr[x\la\cS_n]\Pr[x^*\la\cB(\ket{\psi_x}^{\otimes t(n)})]\cdot\frac{1}{\epsilon(n)}\\
    &=\frac{1}{\epsilon(n)}+\left(1-\frac{1}{\epsilon(n)}\right)\sum_{x,x^*\in\mathsf{Good}}\Pr[x\la\cS_n]\Pr[x^*\la\cB(\ket{\psi_x}^{\otimes t(n)})]
\end{align}
This implies that
\begin{align}
    \frac{\epsilon(n)-p(n)}{p(n)\cdot(\epsilon(n)-1)}\leq \sum_{x,x^*\in\mathsf{Good}}\Pr[x\la\cS_n]\Pr[x^*\la\cB(\ket{\psi_x}^{\otimes t(n)})].
\end{align}
This means that
\begin{align}
    \Pr[
    \begin{array}{ll}
    F(\ket{\psi_{x^*}},\ket{\psi_{x}})\geq \frac{1}{\epsilon(n)}
    \end{array}
        :
    \begin{array}{ll}
       x\la \cS_n  \\
       \ket{\psi_{x}}\la \cC(1^n,x)\\
        x^*\la \cB(\ket{\psi_{x}}^{\otimes t(n)},1^n)\\
        \ket{\psi_{x^*}}\la\cC(1^n,x^*)
    \end{array}
        ]> \frac{\epsilon(n)-p(n)}{p(n)\cdot(\epsilon(n)-1)}=1/\delta(n)
\end{align}
for infinitely many $n\in\N$, which is a contradiction.
Therefore, $(\keygen,\StateGen)$ satisfies
\begin{align}
    \sum_{k,k^*}\Pr[k\la\keygen^*(1^n)]\Pr[k^*\la\cA(\ket{\psi_k}^{\otimes t(n)})]\abs{\bra{\psi_k}\ket{\psi_{k^*}}}^2\leq 1/p(n)
\end{align}
for all polynomials $t$ and infinitely many $n\in\N$.
As shown in \cite{Eprint:MY22}, from $(\keygen,\StateGen)$, we can construct $(\keygen^*,\StateGen^*)$ such that for all polynomials $t$,
\begin{align}
    \sum_{k,k^*}\Pr[k\la\keygen^*(1^n)]\Pr[k^*\la\cA(\ket{\psi_k}^{\otimes t(n)})]\abs{\bra{\psi_k}\ket{\psi_{k^*}}}^2\leq \negl(n),
\end{align}
which completes the proof.
\end{proof}

\subsection{Average-Case Hardness of Learning Mixed State from EFI}
\begin{theorem}\label{thm:Hardness_from_EFI}
    If there exists EFI, then there exists a quantum polynomial-time algorithm with mixed states, which is hard on average in the sense of \cref{def:AH_QLearn}.
\end{theorem}
\begin{remark}
This result works only for "proper" quantum state learning, and it is unclear how to show the hardness of improperly learning mixed states from EFI as far as we understand.
\end{remark}
\cref{thm:Hardness_from_EFI} straightforwardly follows from the definition.
We describe the proof for clarity.

\begin{proof}[Proof of \cref{thm:Hardness_from_EFI}]
    Let $(\keygen,\StateGen)$ be a SV-SI-OWSGs such that 
    $\mathsf{TD}(\psi_{n,x},\psi_{n,x^*})\geq 1-\negl(n)$ for any $x\neq x^*$, where $\psi_{n,x}\la\StateGen(1^n,x)$.
    It is known that EFI is equivalent to SV-SI-OWSGs.
    Therefore, it is sufficient from SV-SI-OWSGs $(\keygen,\StateGen)$ to construct a quantum polynomial-time algorithm $\cC$ such that there exists a sampleable distribution $\cS_n$ over $\cX_n$ and polynomials $\epsilon$ and $\delta$ such that for all polynomials $t$ and all QPT algorithms $\cA$, 
    \begin{align}
        \Pr[
        \begin{array}{ll}
        \mathsf{TD}(\psi_x,\psi_y)\leq 1/\epsilon(n)
        \end{array}
        :
    \begin{array}{ll}
    x\la \cS_n \\
    \psi_x\la \cC(x,1^n)\\
    y\la \cA(\cC,\psi_{x}^{\otimes t(n)},1^n)\\
    \psi_y\la \cC(y,1^n)
    \end{array}
        ]\leq 1/\delta(n)
    \end{align}
    for all sufficiently large $n\in\N$.
    
    Let us describe $\cC$ and $\cS_n$.
    \begin{description}
        \item[The description of $\cS_{n}$:]$ $
        \begin{itemize}
            \item $\cS_n$ outputs the output of $\keygen(1^n)$.
        \end{itemize}
        \item[The description of $\cC$:]$ $
        \begin{itemize}
            \item $\cC(1^n,x)$ outputs the output of $\StateGen(x,1^n)$.
        \end{itemize}
    \end{description}
For contradiction, assume that for all polynomials $\epsilon$ and $\delta$, there exists a QPT learner $\cB$ and a polynomial $t$ such that
\begin{align}
    \Pr[
    \begin{array}{ll}
         &  \mathsf{TD}(\psi_{y},\psi_{x})\leq 1/\epsilon(n)\\
         & h_n\in\cS_n
    \end{array}
    :
    \begin{array}{ll}
         &  x\la\cS_n \\
         & \psi_{x}\la \cC(x,1^n)\\
         & y\la\cB(\cC,\psi_{x}^{\otimes t(n)},1^n)\\
         & \psi_{y}\la \cC(y,1^n)
    \end{array}
    ]> 1/\delta(n)
\end{align}
for infinitely many $n\in\N$.
Now, by using $\cB$, we construct a QPT adversary $\cA$ that breaks the security of the SV-SI-OWSG scheme $\Sigma$ as follows.
\begin{enumerate}
    \item $\cA$ receives sufficiently many copies of $\psi_x$, where $x\la\keygen(1^\secp)$ and $\psi_x\la\StateGen(x)$.
    \item $\cA$ passes them to $\cB$, and receives $x^*$.
    \item $\cA$ sends $x^*$ to the challenger.
\end{enumerate}
Let us define a family $\mathsf{Good}$
\begin{align}
    \mathsf{Good}\seteq \{(x,x^*):\mathsf{TD}(\psi_x,\psi_{x^*})\leq 1/\epsilon(n)\}.
\end{align}
From the definition of $\cB$, we have
\begin{align}
    1/\delta(n)<&\sum_{(x,x^*)\in\mathsf{Good}}\Pr[x\la\keygen(1^n)]\Pr[x^*\la\cB(\psi_x^{\otimes t(n)})]\\
    &=\sum_{x}\Pr[x\la\keygen(1^n)]\Pr[x\la\cB(\psi_x^{\otimes t(n)})]
\end{align}
for infinitely many $n\in\N$.
Here, in the second equation, we have used that for all $x\neq x^*$
\begin{align}
    \mathsf{TD}(\psi_{n,x},\psi_{n,x^*})\geq 1-\negl(n).
\end{align}
Now, we compute the winning probability of the adversary $\cA$ as follows.
\begin{align}
    &\Pr[x\la\cA(\psi_x^{\otimes t(n)}): x\la\keygen(1^n),\psi_x\la\StateGen(1^n,x)]\\
    &=\sum_{x}\Pr[x\la\keygen(1^n)]\Pr[x\la\cA(\psi_x^{\otimes t(n)})]\\
    &=\sum_{x}\Pr[x\la\keygen(1^n)]\Pr[x\la\cB(\psi_x^{\otimes t(n)})]>1/\delta(n)
\end{align}
for infinitely many $n\in\N$.
This contradicts that $\Sigma$ satisfies the security of SV-SI-OWSGs, which completes the proof.
\end{proof}

{\bf Acknowledgements.}
We appreciate Tomoyuki Morimae and Yuki Shirakawa for illuminating discussions.
TH is supported by 
JSPS research fellowship and by JSPS KAKENHI No. JP22J21864.

\ifnum\llncs=1
\bibliographystyle{alpha} 
\bibliography{abbrev3,crypto,reference}
\else
\bibliographystyle{alpha} 
\bibliography{abbrev3,crypto,reference}
\fi

\appendix

\ifnum\cameraready=1
\else
\appendix

	\ifnum\llncs=1
	\newpage
	 	\setcounter{page}{1}
 	{
	\noindent
 	\begin{center}
	{\Large SUPPLEMENTAL MATERIALS}
	\end{center}
 	}
	\setcounter{tocdepth}{2}
	\fi
\fi

\ifnum\cameraready=1
\else
\ifnum\submission=1
\newpage
\setcounter{tocdepth}{1}
\tableofcontents
\else
\fi
\fi

\end{document}